\documentclass[11pt]{amsart}

\usepackage{amssymb,latexsym}
\usepackage{graphicx,epsfig,color}

\newtheorem{thm}{Theorem}[section]

\newtheorem{cor}[thm]{Corollary}
\newtheorem{lem}[thm]{Lemma}
\newtheorem{rem}[thm]{Remark}

\newtheorem{defn}[thm]{Definition}

\numberwithin{equation}{section}

\date{}

\begin{document}

\author[Tulkin H. Rasulov, Elyor B. Dilmurodov]
{Tulkin H. Rasulov, Elyor B. Dilmurodov}
\title[Threshold analysis for a family of $2 \times 2$ operator matrices]
{Threshold analysis for a family of $2 \times 2$ operator matrices}

\maketitle
\begin{center}
{\small Department of Mathematics\\
Faculty of Physics and Mathematics\\
Bukhara State University\\
M. Ikbol str. 11, 200100 Bukhara, Uzbekistan\\
E-mail: rth@mail.ru, elyor.dilmurodov@mail.ru}
\end{center}

\begin{abstract}
We consider a family of $2 \times 2$ operator matrices ${\mathcal
A}_\mu(k),$ $k \in {\Bbb T}^3:=(-\pi, \pi]^3,$ $\mu>0$, acting in
the direct sum of zero- and one-particle subspaces of a Fock space.
It is associated with the Hamiltonian of a system consisting of at
most two particles on a three-dimensional lattice ${\Bbb Z}^3,$
interacting via annihilation and creation operators. We find a set
$\Lambda:=\{k^{(1)},...,k^{(8)}\} \subset {\Bbb T}^3$ and a
critical value of the coupling constant $\mu$ to establish
necessary and sufficient conditions for either
$z=0=\min\limits_{k\in {\Bbb T}^3} \sigma_{\rm ess}({\mathcal
A}_\mu(k))$ ( or $z=27/2=\max\limits_{k\in {\Bbb T}^3} \sigma_{\rm
ess}({\mathcal A}_\mu(k))$ is a threshold eigenvalue or a virtual
level of ${\mathcal A}_\mu(k^{(i)})$ for some $k^{(i)} \in
\Lambda.$
\end{abstract}



\section{\bf Introduction}

Operator matrices are matrices where the entries are linear operators between Banach or
Hilbert spaces, see \cite{CT2008}. One special class of operator matrices are Hamiltonians associated with
the systems of non-conserved number of quasi-particles on a lattice.
In such systems the number of particles can be unbounded as in the case of spin-boson models \cite{HSp95, HS89} or bounded as in the case of "truncated" $ $ spin-boson models \cite{MinSp96, MNR15, OI2018, TR16}.
They arise, for example, in the theory of solid-state physics \cite{Mog91}, quantum field theory \cite{Frid65} and statistical physics \cite{MalMin95, MinSp96}.

The study of systems describing $n$ particles in interaction, without conservation
of the number of particles is reduced to the investigation of the spectral properties of self-adjoint operators acting in the {\it cut subspace} ${\mathcal H}^{(n)}$ of the Fock space, consisting of $r \leq n$ particles \cite{Frid65, MalMin95, MinSp96}.
The perturbation of an operator (the generalized Friedrichs model which has a $2 \times 2$ operator matrix form acting in ${\mathcal H}^{(2)}$), with discrete and essential spectrum has played a considerable role in the study of spectral problems connected
with the quantum theory of fields \cite{Frid65}.

One of the most actively studied objects in operator theory, in
many problems of mathematical physics and other related fields is
the investigation of the threshold eigenvalues and virtual levels
of block operator matrices, in particular, Hamiltonians on a Fock
space associated with systems of non-conserved number of
quasi-particles on a lattice. In the present paper we consider a
family of $2 \times 2$ operator matrices ${\mathcal A}_\mu(k),$ $k
\in {\Bbb T}^3:=(-\pi, \pi]^3,$ $\mu>0$ (so - called generalized
Friedrichs models) associated with the Hamiltonian of a system
consisting of at most two particles on a three-dimensional lattice
${\Bbb Z}^3,$ interacting via creation and annihilation operators.
They are acting in the direct sum of zero-particle and
one-particle subspaces of a Fock space. The main goal of the paper
is to give a thorough mathematical treatment of the spectral
properties of this family in dimension three. More exactly, we
find a set $\Lambda:=\{k^{(1)},...,k^{(8)}\} \subset {\Bbb T}^3$
and prove that for a $i \in \{1,2, \ldots , 8\}$ there is a value
$\mu_i$ of the parameter $\mu$ such that only for $\mu=\mu_i$ the
operator ${\mathcal A}_\mu(\bar{0})$ has a zero-energy resonance,
here $0=\min\sigma_{\rm ess}({\mathcal A}_\mu(\bar{0}))$ and the
operator ${\mathcal A}_\mu(k^{(i)})$ has a virtual level at the
point $z=27/2=\max\sigma_{\rm ess}({\mathcal A}_\mu(k^{(i)}))$,
where $\bar{0}:=(0,0,0)\in {\Bbb T}^3$ and $k^{(i)} \in \Lambda.$
We point out that a part of the results is typical for lattice
models; in fact, they do not have analogues in the continuous case
(because its essential spectrum is half-line $[E; +\infty)$, see
for example \cite{MinSp96}).

 We notice that threshold eigenvalue and virtual level
(threshold energy resonance) of a generalized Friedrichs model
have been studied in \cite{ALR07, MR14, TR11, RD2019}.
The paper \cite{ALM07} is devoted to the threshold analysis
for a family of Friedrichs models under rank one perturbations.
In \cite{ALMM06} a wide class of two-body energy operators $h(k)$
on the $d$-dimensional lattice ${\Bbb Z}^{\rm d}$, ${\rm d} \geq 3,$
is considered, where $k$ is the two-particle quasi-momentum.
If the two-particle Hamiltonian $h(0)$ has either an eigenvalue or a
virtual level at the bottom of its essential spectrum and the
one-particle free Hamiltonians in the coordinate representation
generate positivity preserving semi-groups, then it is shown that
for all nontrivial values $k$, $k \neq 0$, the discrete spectrum of
$h(k)$ below its threshold is non-empty.
These results have been applied to the proof of the existence of
Efimov's effect and to obtain discrete spectrum asymptotics
of the corresponding Hamiltonians.
We note that above mentioned results are discussed only for the
bottom of the essential spectrum. The threshold eigenvalues and
virtual levels of a slightly simpler version of ${\mathcal
A}_\mu(k)$ were investigated in \cite{RD14}, and the
structure of the numerical range are
studied using similar results. In \cite{RT2019}, the essential spectrum of
the family of $3 \times 3$ operator matrices $H(K)$ is described by the spectrum of the family of
$2\times 2$ operator matrices. The results of the present paper are play important role in the
investigations of the operator $H(K)$, see \cite{MR14}.

The plan of this paper is as follows: Section 1 is an introduction
to the whole work. In Section 2, a family of $2 \times 2$ operator matrices are described as bounded self-adjoint
operators in the direct sum of two Hilbert spaces and its spectrum is described.
In Section 3, we discuss some results
concerning threshold analysis of a family of $2 \times 2$ operator matrices.

We adopt the following conventions throughout the present paper.
Let ${\Bbb N},$ ${\Bbb Z},$ ${\Bbb R}$ and ${\Bbb C}$ be the set
of all positive integers, integers, real and complex numbers,
respectively. We denote by ${\Bbb T}^3$ the three-dimensional
torus (the first Brillouin zone, i.e., dual group of ${\Bbb
Z}^3$), the cube $(-\pi,\pi]^3$ with appropriately identified
sides equipped with its Haar measure. The torus ${\Bbb T}^3$ will
always be considered as an abelian group with respect to the
addition and multiplication by real numbers regarded as operations
on the three-dimensional space ${\Bbb R}^3$ modulo $(2 \pi {\Bbb
Z})^3.$

Denote by $\sigma(\cdot),$ $\sigma_{\rm ess}(\cdot)$ and
$\sigma_{\rm disc}(\cdot),$ respectively, the spectrum, the
essential spectrum, and the discrete spectrum of a bounded
self-adjoint operator.

\section{\bf Family of $2\times2$ operator matrices and its spectrum}

Let $L_2({\Bbb T}^3)$ be the Hilbert space of square-integrable
(complex-valued) functions defined on the three-dimensional torus
${\Bbb T}^3$. Denote ${\mathcal H}$ by the direct sum of spaces
${\mathcal H}_0:={\Bbb C}$ and ${\mathcal H}_1:=L_2({\Bbb T}^3)$,
that is, ${\mathcal H}:={\mathcal H}_0 \oplus {\mathcal H}_1$. We
write the elements $f$ of the space ${\mathcal H}$ in the form
$f=(f_0, f_1)$ with $f_0 \in {\mathcal H}_0$ and $f_1 \in
{\mathcal H}_1.$ Then for any two elements $f=(f_0, f_1)$ and
$g=(g_0, g_1)$, their scalar product is defined by
$$
(f,g):=f_0 \overline{g_0}+\int_{{\Bbb T}^3} f_1(t) \overline{g_1(t)} dt.
$$
The Hilbert spaces ${\mathcal H}_0$ and ${\mathcal H}_1$
are zero- and one-particle
subspaces of a Fock space ${\mathcal F}(L_2({\Bbb
T}^3))$ over $L_2({\Bbb T}^3),$ respectively, where
$$
{\mathcal F}(L_2({\Bbb T}^3)):={\Bbb C} \oplus L_2({\Bbb T}^3) \oplus L_2(({\Bbb T}^3)^2) \oplus \cdots \oplus L_2(({\Bbb T}^3)^n)
\oplus \cdots.
$$

In the Hilbert space ${\mathcal H}$ we
consider the following family of $2\times2$ operator matrices
$$
{\mathcal A}_\mu(k):=\left( \begin{array}{cc}
A_{00}(k) & \mu A_{01}\\
\mu A_{01}^* & A_{11}(k)\\
\end{array}
\right),
$$
where $A_{ii}(k): {\mathcal H}_i\to {\mathcal H}_i,$ $i=0,1,$ $k\in {\Bbb T}^3$ and
$A_{01}: {\mathcal H}_1 \to {\mathcal H}_0$ are defined by the rules
$$
A_{00}(k)f_0=w_0(k)f_0,\quad A_{01}f_1=\int_{{\Bbb T}^3} v(t)f_1(t)dt, \quad
(A_{11}(k)f_1)(p)=w_1(k,p)f_1(p).
$$
Here $f_i \in {\mathcal H}_i,$ $i=0,1;$
$\mu>0$ is a coupling constant, the function $v(\cdot)$ is a real-valued analytic function on ${\Bbb T}^3$, the functions
$w_0(\cdot)$ and $w_1(\cdot, \cdot)$ have the form
$$
w_0(k):=\varepsilon(k)+\gamma, \quad
w_1(k,p):=\varepsilon(k)+\varepsilon(k+p)+\varepsilon(p)
$$
with $\gamma \in {\Bbb R}$ and the dispersion function $\varepsilon(\cdot)$
is defined by
\begin{equation}\label{epsilon}
\varepsilon(k):=\sum_{i=1}^3 (1-\cos \, k_i),\,k=(k_1, k_2, k_3) \in
{\Bbb T}^3.
\end{equation}
Under these assumptions the operator matrix ${\mathcal A}_\mu(k)$ is a bounded and
self-adjoint in ${\mathcal H}$.

We remark that the operators $A_{01}$ and $A_{01}^*$ are called
annihilation and creation operators \cite{Frid65}, respectively. In
physics, an annihilation operator is an operator that lowers the
number of particles in a given state by one, a creation operator is
an operator that increases the number of particles in a given state
by one, and it is the adjoint of the annihilation operator.

Let ${\mathcal A}_0(k):={\mathcal A}_\mu(k)\vert_{\mu=0}$. The perturbation ${\mathcal A}_\mu(k)-{\mathcal A}_0(k)$ of the operator ${\mathcal A}_0(k)$ is a
self-adjoint operator of rank 2. Therefore, in accordance with the
invariance of the essential spectrum under the finite rank
perturbations \cite{RS4}, the essential spectrum $\sigma_{\rm ess}({\mathcal A}_\mu(k))$
of ${\mathcal A}_\mu(k)$ fills the following interval on the real axis
$$
\sigma_{\rm ess}({\mathcal A}_\mu(k))=[m(k), M(k)],
$$
where the numbers $m(k)$ and $M(k)$ are defined by
\begin{equation}\label{m(p) and M(p)}
m(k):=\min\limits_{p\in {\Bbb T}^3} w_1(k,p), \quad M(k):=
\max\limits_{p\in {\Bbb T}^3} w_1(k,p).
\end{equation}

\begin{rem}
We remark that the essential spectrum of ${\mathcal A}_\mu(\bar{\pi})$, $\bar{\pi}:=(\pi,\pi,\pi) \in {\Bbb T}^3$ is degenerate to the set consisting of the unique point $\{12\}$ and hence we can not state that the essential spectrum of ${\mathcal A}_\mu(k)$ is absolutely continuous for any $k\in {\Bbb T}^3$.
\end{rem}

For any $\mu>0$ and $k\in {\Bbb T}^3$ we define an analytic function
$\Delta_\mu(k\,; \cdot)$ in ${\Bbb C} \setminus
\sigma_{\rm ess}({\mathcal A}_\mu(k))$ by
\begin{equation*}
\Delta_\mu(k\,; z):=w_0(k)-z-\mu^2 \int_{{\Bbb T}^3} \frac{v^2(t)dt}{w_1(k,t)-z},\,\, z\in {{\Bbb C} \setminus
\sigma_{\rm ess}({\mathcal A}_\mu(k))}.
\end{equation*}

Usually the function $\Delta_\mu(k\,; \cdot)$ is called the
Fredholm determinant associated to the operator matrix ${\mathcal
A}_\mu(k)$.

The following statement establishes connection
between the eigenvalues of the operator  ${\mathcal A}_\mu(k)$ and
zeros of the function $\Delta_\mu(k\,; \cdot)$, see \cite{ALR07, TR11}.

\begin{lem}\label{Lemma 2.1.} For any $\mu>0$ and $k \in {\Bbb T}^3$ the operator ${\mathcal A}_\mu(k)$
has an eigenvalue $z_\mu(k) \in {\Bbb C} \setminus \sigma_{\rm ess}({\mathcal A}_\mu(k))$
if and only if $\Delta_\mu(k\,; z_\mu(k))=0$.
\end{lem}

From Lemma \ref{Lemma 2.1.} it follows that
$$
\sigma_{\rm disc}({\mathcal A}_\mu(k))=\{z\in{\Bbb C} \setminus
\sigma_{\rm ess}({\mathcal A}_\mu(k)):\,\Delta_\mu(k\,; z)=0\}.
$$

Since the function $\Delta_\mu(k\,; \cdot)$ is a monotonically decreasing function on $(-\infty; m(k))$ and $(M(k); +\infty)$,
for $\mu>0$ and $k \in {\Bbb T}^3$ the operator ${\mathcal A}_\mu(k)$ has no more than 1 simple eigenvalue in $(-\infty; m(k))$ and $(M(k); +\infty)$.

Let $\Lambda:=\{k=(k_1, k_2, k_3): k_i \in \{-2\pi/3, 2\pi/3 \},
i=1,2,3 \}.$ Since the set $\Lambda \subset {\Bbb T}^3$ consists 8
points for a convenience we rewrite the set $\Lambda$ as
$\Lambda=\{k^{(1)}, k^{(2)},\ldots, k^{(8)}\}$.

It is easy to verify that the function $w_1(\cdot,\cdot)$ has a
non-degenerate minimum at the point $(\bar{0},\bar{0}) \in ({\Bbb
T}^3)^2,$ $\bar{0}:=(0,0,0)$ and has non-degenerate maximum at the
points of the form $(k^{(i)}, k^{(i)}) \in ({\Bbb T}^3)^2$,
$i=1,\ldots,8$, such that
$$
\min\limits_{k,p\in {\Bbb T}^3}w_1(k,p)=w_1(\bar{0},\bar{0})=0,\quad
\max\limits_{k,p\in {\Bbb T}^3} w_1(k,p)=w_2(k^{(i)}, k^{(i)})=27/2, \quad i=1,\ldots,8.
$$

Simple calculations show that
\begin{align*}
& \sigma_{\rm ess}({\mathcal A}_\mu(\bar{0}))=[0; 12];\\
& \sigma_{\rm ess}({\mathcal A}_\mu(k^{(i)}))=[\frac{15}{2}; \frac{27}{2}], \quad i=1,\ldots,8.
\end{align*}

Therefore,
$$
\min\limits_{k \in {\Bbb T}^3}\sigma_{\rm ess}({\mathcal A}_\mu(k))=0,\quad
\max\limits_{k \in {\Bbb T}^3}\sigma_{\rm ess}({\mathcal A}_\mu(k))=\frac{27}{2}.
$$

\section{\bf Threshold eigenvalues and virtual levels.}

In this Section we prove that for any $i \in \{1,\ldots,8\}$ there
is a value $\mu_i$ of the parameter (coupling constant) $\mu$ such
that only for $\mu=\mu_i$ the operator ${\mathcal A}_\mu(\bar{0})$
has a virtual level at the point $z=0$ (zero-energy resonance) and
the operator ${\mathcal A}_\mu(k^{(i)})$ has a virtual level at
the point $z=27/2$ under the assumption that $v(\bar{0}) \neq 0$
and $v(k^{(i)}) \neq 0$. For the case $v(\bar{0})=0$ and
$v(k^{(i)})=0$ we show that the number $z=0$ ($z=27/2$) is a
threshold eigenvalue of ${\mathcal A}_\mu(\bar{0})$ (${\mathcal
A}_\mu(k^{(i)})$).

Denote by $C({\Bbb T}^3)$ and $L_1({\Bbb T}^3)$ the Banach spaces of continuous and integrable functions on ${\Bbb T}^3$, respectively.

\begin{defn}\label{Definition 3.1.}
Let $\gamma \neq 0.$ The operator ${\mathcal A}_\mu(\bar{0})$ is said to have a virtual level at $z=0$
$($or zero-energy resonance$)$, if the number $1$
is an eigenvalue of the integral operator
$$
(G_{\mu}\psi)(p)=\frac{\mu^2 v(p)}{2\gamma} \int_{{\Bbb T}^3} \frac{v(t)\psi(t)dt}{\varepsilon(t)},\quad  \psi\in C({\Bbb T}^3)
$$
and the associated eigenfunction $\psi(\cdot)$ $($up to constant factor$)$ satisfies the condition $\psi(\bar{0})\neq 0.$
\end{defn}

\begin{defn}\label{Definition 3.2.}
Let $\gamma \neq 9$ and $i \in \{1,\ldots,8\}.$ The operator ${\mathcal A}_\mu(k^{(i)})$ is said to have a virtual level at $z=27/2$, if the number $1$ is an eigenvalue of the integral operator
$$
(G^{(i)}_{\mu}\varphi)(p)=\frac{\mu^2 v(p)}{\gamma-9} \int_{{\Bbb T}^3} \frac{v(t)\varphi(t)dt}
{\varepsilon(k^{(i)}+t)+\varepsilon(t)-9},\quad  \varphi\in C({\Bbb T}^3)
$$
and the associated eigenfunction $\varphi(\cdot)$ $($up to constant factor$)$
satisfies the condition $\varphi(k^{(i)})\neq 0.$
\end{defn}

Using the extremal properties of the function $\varepsilon(\cdot)$, and the Lebesgue dominated convergence theorem
we obtain that there exist the positive finite limits
\begin{align*}
& \lim\limits_{z\to-0}\int_{{\Bbb T}^3} \frac{v^2(t)dt}{\varepsilon(t)-z}=\int_{{\Bbb T}^3} \frac{v^2(t)dt}{\varepsilon(t)};\\
& \lim\limits_{z\to 9+0}\int_{{\Bbb T}^3} \frac{v^2(t)dt}{z-\varepsilon(k^{(i)}+t)-\varepsilon(t)}=\int_{{\Bbb T}^3} \frac{v^2(t)dt}{9-\varepsilon(k^{(i)}+t)-\varepsilon(t)}.
\end{align*}

For the next investigations we define the following quantities
\begin{align*}
& \mu_l(\gamma):=\sqrt{2 \gamma} \left(\int_{{\Bbb T}^3} \frac{v^2(t)dt}{\varepsilon(t)} \right)^{-1/2}\,\, \mbox{for}\,\, \gamma>0;\\
& \mu_r^{(i)}(\gamma):=\sqrt{9-\gamma} \left(\int_{{\Bbb T}^3} \frac{v^2(t)dt}{9-\varepsilon(k^{(i)}+t)-\varepsilon(t)} \right)^{-1/2}\,\, \mbox{for}\,\, \gamma<9,\,\, i=1,\ldots,8.
\end{align*}

Let $\gamma_i \in (0; 9)$ be an unique solution of $\mu_l(\gamma)=\mu_r^{(i)}(\gamma).$ It follows immediately that
$$
\gamma_i:=9 \left(2\int_{{\Bbb T}^3} \frac{v^2(t)dt}{9-\varepsilon(k^{(i)}+t)-\varepsilon(t)}
+\int_{{\Bbb T}^3} \frac{v^2(t)dt}{\varepsilon(t)} \right)^{-1}\, \int_{{\Bbb T}^3} \frac{v^2(t)dt}{\varepsilon(t)}.
$$

In the following we compare the values of $\mu_l (\gamma)$ and
$\mu_r (\gamma)$ depending on $\gamma \in(0;9)$.
\begin{rem}
Let $i \in \{1,\ldots,8\}$.
By the definition of the quantities $\mu_l(\gamma)$ and $\mu_r^{(i)}(\gamma)$ one can conclude that\\
if $\gamma\in (0; \gamma_i),$ then $\mu_l(\gamma)<\mu_r^{(i)}(\gamma);$\\
if $\gamma=\gamma_i$, then $\mu_l(\gamma)=\mu_r^{(i)}(\gamma);$\\
if $\gamma\in (\gamma_i; 9),$ then $\mu_l(\gamma)>\mu_r^{(i)}(\gamma).$
\end{rem}

From the Definition \ref{Definition 3.1.} (resp. \ref{Definition
3.2.}) we obtain that the number $1$ is an eigenvalue of $G_{\mu}$
$($resp. $G^{(i)}_{\mu})$ if and only if $\mu=\mu_l(\gamma)$
$($resp. $\mu=\mu_r^{(i)}(\gamma))$.

We notice that in the Definition \ref{Definition 3.2.}, the requirement of the presence of an eigenvalue $1$ of $G^{(i)}_{\mu}$ corresponds to the existence
of a solution of the equation ${\mathcal A}_\mu(k^{(i)})f=(27/2) f$
and the condition $\psi(k^{(i)}) \neq 0$ implies that the solution $f=(f_0, f_1)$ of this equation does not belong to ${\mathcal H}$.
More exactly, if the operator ${\mathcal A}_\mu(k^{(i)})$ has a virtual level at $z=27/2$, then the
vector-function $f=(f_0, f_1)$, where
\begin{equation}\label{eigenvector}
f_0={\rm const} \neq 0,\quad f_1(q)=-\frac {\mu
v(q)f_0}{\varepsilon(k^{(i)}+q)+\varepsilon(q)-9} ,
\end{equation}
satisfies the equation ${\mathcal A}_\mu(k^{(i)})f=(27/2) f$ and $f_1\in L_1({\Bbb T}^3)\setminus L_2({\Bbb T}^3)$
(see assertion (i) of Theorem \ref{Main Theorem 1}).

 If the number $z=27/2$ is an eigenvalue of the operator ${\mathcal A}_\mu(k^{(i)})$ then the vector-function $f=(f_0,f_1),$ where $f_0$ and $f_1$ are defined in \eqref{eigenvector}, satisfies the equation ${\mathcal A}_\mu(k^{(i)})f=(27/2) f$ and $f_1\in L_2({\Bbb T}^3)$
(see assertion (ii) of Theorem \ref{Main Theorem 1}).

The same assertions are true for the operator ${\mathcal A}_\mu(\bar{0})$ at the point $z=0.$

Henceforth, we shall denote by $C_1, C_2, C_3$ different positive numbers and
for each $\delta>0,$ the notation $U_\delta(p_0)$ is used for the
$\delta-$neighborhood of the point $p_0\in {\Bbb T}^3:$
$$
U_\delta(p_0):=\{p\in {\Bbb T}^3: |p-p_0|<\delta\}.
$$

Now we formulate the first main result of the paper.

\begin{thm}\label{Main Theorem 1} Let $\gamma<9$ and $i \in \{1,\ldots,8\}$.\\
{\rm (i)} The number $z=27/2$ is an eigenvalue of the operator
${\mathcal A}_\mu(k^{(i)})$ if and only if
$\mu=\mu_r^{(i)}(\gamma)$ and $v(k^{(i)})=0;$\\
{\rm (ii)} The operator ${\mathcal A}_\mu(k^{(i)})$ has a virtual
level at the point $z=27/2$ if and only if
$\mu=\mu_r^{(i)}(\gamma)$ and $v(k^{(i)}) \neq 0.$
\end{thm}

\begin{proof}
Suppose $\gamma<9$ and $i \in \{1,\ldots,8\}$.

(i) "Only If Part". Let the number $z=27/2$ be an eigenvalue of ${\mathcal A}_\mu(k^{(i)})$ and $f=(f_0, f_1) \in {\mathcal H}$
be an associated eigenvector. Then $f_0$ and $f_1$ are satisfy the system of equations
\begin{align}\label{system of equation}
& (\gamma-9)f_0+\mu \int_{{\Bbb T}^3} v(t)f_1(t)dt=0; \nonumber\\
& \mu v(p)f_0  +(\varepsilon(k^{(i)}+p)+\varepsilon(p)-9)f_1(p)=0.
\end{align}

This implies that $f_0$ and $f_1$ are of the form
\eqref{eigenvector} and the first equation of system \eqref{system
of equation} yields $\Delta_{\mu}(k^{(i)}\,; 27/2)=0$, therefore,
$\mu=\mu_r^{(i)}(\gamma)$.

Now let us show that $f_1 \in L_2({\Bbb T}^3)$ if and only if $v(k^{(i)})=0.$
Indeed, if $v(k^{(i)})=0$ (resp. $v(k^{(i)})\neq 0$), from analyticity of the function $v(\cdot)$ it follows that there exist
$C_1,C_2,C_3>0$, $\theta_i \in {\Bbb N}$ and $\delta>0$ such that
\begin{equation}\label{two-sided estimate for v}
C_1|p-k^{(i)}|^{\theta_i} \leq |v(p)| \leq C_2|p-k^{(i)}|^{\theta_i},\quad p \in U_{\delta}(k^{(i)}),
\end{equation}
respectively
\begin{equation}\label{estimate for v}
|v(p)| \geq C_3,\quad p \in {\Bbb T}^3 \setminus U_{\delta}(k^{(i)}).
\end{equation}

Since the function $\varepsilon(k^{(i)}+p)+\varepsilon(p)$ has an unique non-degenerate maximum at the point $k^{(i)} \in {\Bbb T}^3$
there exist $C_1,C_2,C_3>0$ and $\delta>0$ such that
\begin{equation}\label{two-sided estimate for epsilon}
C_1|p-k^{(i)}|^2 \leq |\varepsilon(k^{(i)}+p)+\varepsilon(p)-9| \leq C_2|p-k^{(i)}|^2,\quad p \in U_{\delta}(k^{(i)}),
\end{equation}
\begin{equation}\label{estimate for epsilon}
|\varepsilon(k^{(i)}+p)+\varepsilon(p)-9| \geq C_3,\quad p \in {\Bbb T}^3 \setminus U_{\delta}(k^{(i)}).
\end{equation}

We have
\begin{align}\label{equality for int of f1}
\int_{{\Bbb T}^3} |f_1(t)|^2dt&=\mu^2 |f_0|^2 \int_{U_{\delta}(k^{(i)})} \frac {v^2(t)dt}{(\varepsilon(k^{(i)}+t)+\varepsilon(t)-9)^2}\nonumber\\
&+\mu^2 |f_0|^2 \int_{{\Bbb T}^3\setminus U_{\delta}(k^{(i)})} \frac {v^2(t)dt}{(\varepsilon(k^{(i)}+t)+\varepsilon(t)-9)^2}.
\end{align}

Let $v(k^{(i)})=0$. Then by \eqref{two-sided estimate for v} and \eqref{two-sided estimate for epsilon}
for the first summand on the right-hand side of \eqref{equality for int of f1} we have
$$
\int_{U_{\delta}(k^{(i)})} \frac
{v^2(t)dt}{(\varepsilon(k^{(i)}+t)+\varepsilon(t)-9)^2} \leq C_1
\int_{U_{\delta}(k^{(i)})} \frac {|t-k^{(i)}|^{2\theta_i}
dt}{|t-k^{(i)}|^4}<+\infty.
$$
It follows from the continuity of $v(\cdot)$ on a compact set
${\Bbb T}^3$ and \eqref{estimate for epsilon} that
$$
\int_{{\Bbb T}^3\setminus U_{\delta}(k^{(i)})} \frac {v^2(t)dt}{(\varepsilon(k^{(i)}+t)+\varepsilon(t)-9)^2} \leq C_1 \int_{{\Bbb T}^3\setminus U_{\delta}(k^{(i)})}dt<+\infty.
$$
So, in this case $f_1 \in L_2({\Bbb T}^3).$

For the case $v(k^{(i)}) \neq 0$ there exsist the numbers $\delta
>0$ and $C_1 >0$ such that $|v(p)| \geq C_1$ for any $p \in U_\delta (k^{(i)})$. Then from \eqref{two-sided estimate for epsilon} we obtain
$$
\int_{{\Bbb T}^3} |f_1(t)|^2dt \geq C_1 \int_{U_{\delta}(k^{(i)})} \frac {dt}{|t-k^{(i)}|^4}=+\infty.
$$

Therefore, $f_1 \in L_2({\Bbb T}^3)$ if and only if $v(k^{(i)})=0$.

"If Part". Suppose that $\mu=\mu_r^{(i)}(\gamma)$ and
$v(k^{(i)})=0.$ It is easy to verify that the vector-function
$f=(f_0,f_1)$ with $f_0$ and $f_1$ defined in \eqref{eigenvector}
satisfies the equation ${\mathcal A}_\mu(k^{(i)})f=(27/2) f$. We
proved above that if $v(k^{(i)})=0$, then $f_1 \in L_2({\Bbb
T}^3)$.

(ii) "Only If Part". Suppose that the operator ${\mathcal A}_\mu(k^{(i)})$ has a virtual level
at $z=27/2.$ Then by Definition \ref{Definition 3.2.} the equation
\begin{equation}\label{def of virtual level}
\varphi(p)=\frac{\mu^2 v(p)}{\gamma-9} \int_{{\Bbb T}^3} \frac{v(t)\varphi(t)dt}
{\varepsilon(k^{(i)}+t)+\varepsilon(t)-9},\quad  \varphi\in C({\Bbb T}^3)
\end{equation}
has a nontrivial solution $\varphi\in C({\Bbb T}^3)$, which satisfies the condition
$\varphi(k^{(i)}) \neq 0.$

This solution is equal to the function $v(p)$ (up to a constant factor)
and hence
$$
\Delta_\mu(k^{(i)}, 27/2)=\gamma-9-\mu^2 \int_{{\Bbb T}^3} \frac{v^2(t)dt}{\varepsilon(k^{(i)}+t)+\varepsilon(t)-9}=0,
$$
that is, $\mu=\mu_r^{(i)}(\gamma).$

"If Part". Let now $\mu=\mu_r^{(i)}(\gamma)$ and $v(k^{(i)}) \neq 0.$
Then the function $v \in C({\Bbb T}^3)$ is a solution of \eqref{def of virtual level}, and consequently, by Definition \ref{Definition 3.2.} the operator
${\mathcal A}_\mu(k^{(i)})$ has a virtual level at $z=27/2.$
\end{proof}

The following result may be proved in much the same way as Theorem \ref{Main Theorem 1}.

\begin{thm}\label{Main Theorem 2} Let $\gamma>0$.\\
{\rm (i)} The operator ${\mathcal A}_\mu(\bar{0})$ has an zero
eigenvalue if and only if  $\mu=\mu_l(\gamma)$ and
$v(\bar{0})=0;$\\
{\rm (ii)} The operator ${\mathcal A}_\mu(\bar{0})$ has a  zero
energy resonance if and only if $\mu=\mu_l(\gamma)$ and
$v(\bar{0}) \neq 0.$
\end{thm}

Since $\mu_l(\gamma_i)=\mu_r^{(i)}(\gamma_i),$ setting $\mu_i:=\mu_l(\gamma_i),$ from Theorems \ref{Main Theorem 1} and \ref{Main Theorem 2} we obtain the following

\begin{cor}
Let $\gamma \in (0; 9)$ and $i \in \{1,\ldots,8\}.$\\
{\rm (i)} The operator ${\mathcal A}_\mu(\bar{0})$ has a zero
eigenvalue and
the number $z=27/2$ is an eigenvalue of ${\mathcal A}_\mu(k^{(i)})$ iff $\mu=\mu_i$ and $v(\bar{0})=v(k^{(i)})=0;$\\
{\rm (ii)} The operator ${\mathcal A}_\mu(\bar{0})$ has
zero-energy resonance and the operator ${\mathcal A}_\mu(k^{(i)})$
has a virtual level at the point $z=27/2$
iff $\mu=\mu_i,$ $v(\bar{0}) \neq 0$ and $v(k^{(i)}) \neq 0;$\\
{\rm (iii)} The operator ${\mathcal A}_\mu(\bar{0})$ has a zero
eigenvalue and the operator ${\mathcal A}_\mu(k^{(i)})$ has a
virtual level at the point $z=27/2$
iff $\mu=\mu_i,$ $v(\bar{0}) = 0$ and $v(k^{(i)}) \neq 0;$\\
{\rm (iv)} The operator ${\mathcal A}_\mu(\bar{0})$ has a
zero-energy resonance and the number $z=27/2$ is an eigenvalue of
${\mathcal A}_\mu(k^{(i)})$ iff $\mu=\mu_i,$ $v(\bar{0}) \neq 0$
and $v(k^{(i)}) = 0.$
\end{cor}

Next we will consider some applications of the results. Denote by ${\mathcal H}_2:=L_2^{\rm s} (({\Bbb T}^3)^2)$ the
Hilbert space of square integrable (complex) symmetric functions
defined on $({\Bbb T}^3)^2.$ In the Hilbert space ${\mathcal H}_1 \oplus {\mathcal H}_2$ we
consider a $2\times2$ operator matrix
$$
{\mathcal A}_\mu:=\left( \begin{array}{cc}
A_{11} & \sqrt{2} \mu A_{12}\\
\sqrt{2}\mu A_{12}^* & A_{22}\\
\end{array}
\right),
$$
where $A_{ij}: {\mathcal H}_j \to {\mathcal H}_i,$ $i=1,2$ are defined by the rules
\begin{align*}
& (A_{11}f_1)(k)=w_1(k)f_1(k),\quad (A_{12}f_2)(k)=\int_{{\Bbb T}^3} v(t)f_2(k,t)dt,\\
& (A_{22}f_2)(k,p)=w_1(k,p)f_2(k,p)\quad f_i \in {\mathcal H}_i, \quad i=1,2.
\end{align*}
Here $A_{12}^*: {\mathcal H}_1 \to {\mathcal H}_2$ denotes the adjoint operator to $A_{12}$
and
$$
(A_{12}^*f_1)(k,p)=\frac{1}{2} (v(k)f_1(p)+v(p)f_1(k)), \quad f_1 \in {\mathcal H}_1.
$$

Under these assumptions the operator ${\mathcal A}_\mu$ is bounded
and self-adjoint.

The main results of the present paper plays crucial role in the study of the spectral properties of the operator matrix ${\mathcal A}_\mu$.
In particular, the essential spectrum of ${\mathcal A}_\mu$ can be described via the spectrum of ${\mathcal A}_\mu(k)$ the following equality holds
$$
\sigma_{\rm ess}({\mathcal A}_\mu)=[0; 27/2] \cup \bigcup_{k \in {\Bbb T}^3} \sigma_{\rm disc}({\mathcal A}_\mu(k)).
$$
Since the operator ${\mathcal A}_\mu(k)$ has at most 2 simple eigenvalues,
the set $\sigma_{\rm ess}({\mathcal A}_\mu)$ consists at least one and at most three bounded closed intervals,
for similar results see \cite{TR16}.

Using Theorems \ref{Main Theorem 1} and \ref{Main Theorem 2} one can investigate \cite{TR11} the number of eigenvalues of ${\mathcal A}_\mu$
and find its discrete spectrum asymptotics.

We note that the case
$$
v(p)=\sqrt{\mu}=const, \quad w_1
(k,p)=\varepsilon(k)+\varepsilon(\frac{1}{2} (k+p))+\varepsilon(p)
$$
is studied in \cite{RD2019}, and it is shown that the bounds
$\min\limits_{k\in {\Bbb T}^3} \sigma_{ess}({\mathcal
A}_\mu(\bar{0}))$ and $\max\limits_{k\in {\Bbb T}^3}
\sigma_{ess}({\mathcal A}_\mu(\bar{\pi}))$ are only virtual
levels. This paper generalizes the results of the paper
\cite{RD2019} and it is proved that these bounds are threshold
eigenvalues or virtual levels depending on the values of the
function $v(\cdot)$.

\end{document}